\documentclass[letterpaper, 10 pt, conference]{ieeeconf}

\IEEEoverridecommandlockouts                              
\usepackage{amsmath,amsfonts,amssymb}
\usepackage{graphicx}
\usepackage{subfig}

\newcommand{\tr}{\mbox{\rm tr}}    

\newfont{\gothic}{eufm10 scaled\magstep0}

\newcommand{\rr}{\mbox{$\mathbb R$}}
\newcommand{\nn}{\mbox{$\mathbb N$}}

\newtheorem{theorem}{Theorem}[section]
\newtheorem{lemma}[theorem]{Lemma}

\newtheorem{definition}[theorem]{Definition}
\newtheorem{make_remark}[theorem]{Remark}

\begin{document}

\title{Observer design for position and velocity bias estimation\\from a single direction output}

\author{Florent Le~Bras, Tarek~Hamel, Robert~Mahony, Claude Samson,
\thanks{F. Le~Bras is with the French Direction G\'en\'erale de l'Armement (Technical Directorate), Bagneux, France, {\tt florent.le-bras@polytechnique.org}.}
\thanks{T. Hamel is with I3S UNS-CNRS, Nice-Sophia Antipolis, France, {\tt thamel@i3s.unice.fr}.}
\thanks{R. Mahony is with the School of Engineering, Australian National University, ACT, 0200, Australia, {\tt Robert.Mahony@anu.edu.au}}
\thanks{C. Samson is with INRIA and  I3S UNS-CNRS, Sophia Antipolis, France, {\tt claude.samson@inria.fr, csamson@i3s.unice.fr}. }}

\maketitle


\begin{abstract}
This paper addresses the problem of estimating the position of an object moving in $\rr^n$ from direction and velocity measurements. After addressing observability issues associated with this problem, a nonlinear observer is designed so as to encompass the case where the measured velocity is corrupted by a constant bias. Global exponential convergence of the estimation error is proved under a condition of persistent excitation upon the direction measurements. Simulation results illustrate the performance of the observer.


\end{abstract}

\section{Introduction}

There is is rich literature in vision based pose estimation driven by advances in the structure from motion in the field of computer vision \cite{2010_Haeming}.
Most of the recent structure from motion algorithms are formulated as an optimisation problem over a set of selected images  \cite{1999_Triggs_iccv}, however, recent work has emphasised the importance of considering motion models and filtering techniques \cite{2012_Strasdat} for a class of important problems.
Recursive filtering methods for vision based structure from motion and pose estimation themselves have a rich history primarily associated with stochastic filter design such as EKF, unscented filters and particle filters \cite{1989_Matthies,1996_Soatto_TAC,2007_Armesto,2008_Civera}.
A comparison of EKF and particle filter algorithms for vision based SLAM is available in \cite{2006_Bekris}.
Although nonlinear observer design does not provide a stochastic interpretation of the state estimate they hold the promise to handle the non-linearity of the vision pose estimation problem in a robust and natural manner \cite{RM_2009_Baldwin.icra}.
Ghosh and subsequent authors consider non-linear observers on the class of perspective systems \cite{RehGho2003_TAC,2004_Abdursul,Gho2004Aut,2006_Aguiar.TAC,Dahl20101829}, that is systems with output in a projective space obtained as a quotient of the state space.
Perspective outputs $y(x)$ are of the form
\[
y^P=(\frac{x_1}{x_n}, \ldots \frac{x_{n-1}}{x_n},1)
\]
and correspond to the nonlinear projection along rays through the origin onto an affine image plane perpendicular to the focal axis.
The output representation is attractive in that it corresponds to the normal representation of vision data for perspective cameras.
Indeed, there are a number of works that consider filtering for $y^P$ directly, rather than estimating the camera position \cite{dixon2003range,de2007line,dani2012globally}, corresponding to image tracking.
Although significant work has been based on this output representation, it tends to lead to complex observer and filter design and difficult analysis
\cite{RehGho2003_TAC,2004_Abdursul,Gho2004Aut,2006_Aguiar.TAC,Dahl20101829}.
An additional question of importance concerns the rate of convergence of an observer and recent work has addressed this question in the context of controlling the camera motion to improve observability of the problem and increase the rate of convergence of the observer \cite{2014_Spica}.

The present paper contributes further to the field of nonlinear observer design for systems with direction outputs.
The key contribution that we make is the development of an elegant and rigorous stability analysis for a simple filter design.
The filter is designed for a single bearing measurement and relies on the motion of the camera to generate persistence of excitation of the innovation in order to guarantee global asymptotic convergence.
Rather than using the perspective outputs favoured in previous papers we use direction outputs
\[
y = x /|x| = \frac{y^P}{|y^P|}
\]
corresponding to projection onto a virtual spherical image plane and differing from perspective outputs only in the scaling.
The two formulations are essentially equivalent from a systems perspective in the region where perspective outputs are defined. However, we believe that the direction output representation contributes to the simplicity of the observer proposed in the present paper.
We characterise the rate of convergence of the filter in terms of the persistence of excitation property.
We then consider the case when the measurement of velocity of the camera is perturbed by an unknown bias.
To the authors knowledge, this problem has not been considered in the nonlinear observer literature.
We provide a rigorous proof of the global asymptotic stability of the observer state for this case by exploiting a novel state transformation.
The simulations provided demonstrate the performance of the filter.

The paper is organised along five sections. Following the present introduction, Section~\ref{sec:prelim} introduces the system under consideration and points out observability properties attached to it.
Section~\ref{sec:observer} develops the main results of the paper.
Section~\ref{sec:simulations} present a few illustrative simulations. Concluding remarks are provided in Section~\ref{sec:conclusions}.

\section{Problem description}\label{sec:prelim}

The system considered is the kinematics of an object moving in $\mathbb{R}^{n}$
\begin{align}
\dot x &=v +a \label{eq:Kinematics} \\
y &=\frac{x}{|x|} \in S^{n-1} \label{eq:output}
\end{align}
where $v \in \rr^n$ is the velocity of the object and $a \in \rr^n$ represents any unknown bias.  Let $S^{n-1}$ denote the unit sphere, the space of measurements $y \in \rr^n$ such that\footnote{$|.|$ stands for the Euclidean norm of vectors and $||.||$ is the induced matrix norm.} $|y|=1$. An example of such a measurement is  the bearing in $S^{2}$  obtained from a camera looking at a moving object.

In most applications the unknown velocity $a \in \mathbb{R}^{n}$ (with $n=3$) represents the velocity of the fluid in which evolves the moving object or/and any bias that affects the measurement of $v$.
In this paper, for the sake of generality, we consider an arbitrary dimension $n\geq2$.
We emphasize that the value of $y \in S^{n-1}$ and $v \in \mathbb{R}^n$ must be known at all times.


\subsection{Observability analysis}\label{:sub:notation}

We first give a general observability criterion. The following persistency of excitation condition will then yield
an observability result for system (\ref{eq:Kinematics}-\ref{eq:output}).
\begin{definition}\label{def:persis}
  The direction  $y \in
  S^{n-1}$, is called \emph{persistently exciting} if there
  exist $\delta>0$ and $0<\mu<\delta$  such that for all $t$
 \begin{equation}
    \label{eq:persistence1}
    \int_{t}^{t+\delta}
      \pi_{y(\tau)} \mathrm{d}\tau \geq \mu I, \; \mbox{with } \pi_y=(I-yy^{\top})
  \end{equation}
\end{definition}
For future use, note that \eqref{eq:persistence1} is equivalent to
  \begin{equation}
 \label{eq:persistence2}
 \forall b \in S^{n-1}~:~\int_{t}^{t+\delta}
     |\pi_{y(\tau)}b|^2 \mathrm{d}\tau \geq \mu,
  \end{equation}
Another characterization of persistent of excitation, in terms of the property that the time-derivative of $\dot{y}$ must satisfy, is pointed out in the following lemma
\begin{lemma}
  \label{lem:PE2}
Assume that $\dot{y}(t)$ is uniformly continuous and bounded, then relation \eqref{eq:persistence1} (respectively \eqref{eq:persistence2}) is equivalent to:
\begin{equation}\label{eq:persistence4}
\exists \varepsilon>0, \; \exists  \tau \in [t,t+\delta] \mbox{ such that } |\dot{y}(\tau)| \geq \varepsilon
\end{equation}
\end{lemma}
~\\
\begin{proof}The proof of this lemma is given in Appendix \label{Plem:PE2}
\end{proof}
Note that the uniform continuity and boundedness of $\dot{y}$ is automatically granted when $v$ is itself uniformly continuous and bounded and $|x|$ is lower bounded by some positive number.

Recall that two different points $x_1^0,x_2^0\in \mathbb{R}^n $ are said distinguishable, if there exists an input
$v(t) \in \rr^n$ and a time $t_1$ such that for solutions
$x_1(t),x_2(t)$ of \eqref{eq:Kinematics} with $x_1(0)=x_1^0$, $x_2(0)=x_2^0$ we have $y(x_1(t_1))\not=y(x_2(t_1))$.
Equivalently, in this case one says that the admissible input distinguishes the
two initial states, and also that two initial states of system~(\ref{eq:Kinematics}-\ref{eq:output}) are indistinguishable if they are
not distinguished by any admissible input.
\begin{definition}
A system is called \emph{strongly observable} if all pairs of distinct initial states are distinguished by all
admissible inputs. It is called \emph{weakly observable} if every pair of distinct
initial states is distinguished by at least one admissible input.
\end{definition}

Reasons to differentiate between strong observability and weak
observability are well explained in the non-linear control literature. For complementary details on this subject we refer the reader
to a classical work by Sussmann~\cite{Sussmann}.

\begin{lemma}
The system (\ref{eq:Kinematics}-\ref{eq:output}) complemented with the equation $\dot{a}=0$, with $X=\left( \begin{array}{c} x \\ a \end{array} \right)$ as the system state vector, $v$ as the system input, and $y$ as the system output, is weakly observable but not strongly observable.
\end{lemma}
\begin{proof}
Choose, for instance, the input $v(t)=(\cos(t),\sin(t),0 \ldots 0)^T$. The solutions to the system are then given by $x(t)=x(0)+(\sin(t)+at,-\cos(t)+at,0 \ldots 0)^T$ and one easily verifies that $y_1(t)=y_2(t)$, $\forall t$, implies that $x_1(0)=x_2(0)$. This establishes the weak observability property of the system. Note also that the chosen input renders both outputs $y_1(t)$ and $y_2(t)$ persistently exciting in the sense of the definition \eqref{def:persis}. On the other hand, one verifies that, if the input $v$ is constant, then initial states $x_1(0)=k_1(v+a)$ and $x_2(0)=k_2(v+a)$, with $k_1$ and $k_2$ denoting arbitrary positive numbers, can not be distinguished because $y_1$ and $y_2$ are constant and equal in this case. This proves that the system is not strongly observable.
\end{proof}
The weak observability property of the system justifies the introduction of the persistence condition evoked previously to characterize "good" outputs (produced by "good" inputs) yielding a property of "uniform" observability that renders the state-observation problem addressed in the next section well-posed.

\section{Observer design}\label{sec:observer}
The problem of state observation refers to the design of an algorithm that allows one to recover actual state values from the observation of previous outputs.
We start by the observer design for the classical situation addressed in the literature where the unknown constant velocity bias $a$ is equal to zero. The situation when this term is different from zero and unknown {\em a priori} is addressed subsequently.

  \begin{lemma}
  \label{th:complementary_convergence}
  Consider the system~(\ref{eq:Kinematics}-\ref{eq:output}) and
  the  following observer:
  \begin{equation} \label{filter}
  \dot{\hat{x}}_1=v -k \pi_y \hat{x}_1, \hat{x}_1(0)=\hat{x}_1^0 \in \mathbb{R}^n \mbox{ and } k>0
  \end{equation}
 Assume that $a \equiv 0$, $x$ is bounded and never crosses zero, so that the output $y$ is always well defined. Let $\tilde{x}_1=x-\hat{x}_1 \in \mathbb{R}^n$ denote the estimation error.
  If $v(t) \in \mathbb{R}^n$ is bounded and such that the measured direction  $y(t)$ is
 persistently exciting, then the equilibrium $\tilde{x}_1=0$ is Uniformly Globally
  Exponentially Stable (UGES).
\end{lemma}
\begin{proof}
Differentiating $\tilde{x}$ and using \eqref{eq:Kinematics} and \eqref{filter} one gets the following linear time varying system:
\[
    \dot{\tilde{x}}_1 =k \pi_y \hat{x}=-k \pi_y \tilde{x}_1, \]

Using the assumption of persistent excitation characterized by relation \eqref{eq:persistence1} a direct application of \cite[Lemma 5]{loria2002uniform} proves that $\tilde{x}_1$ is UGES. More explicitly, one verifies that the transition matrix $\Phi$ associated with the above system satisfies
\begin{equation} \label{transition-bounds}
\exp^{-k(t-\tau)}\leq ||\Phi(t,\tau)||\leq \exp^{-\gamma (t-\tau)},
\end{equation}
where $\gamma=\frac{\mu k}{\delta(1+k^2\delta)^2}$.
 \end{proof}

The interest of this result lies in the extreme simplicity of the observer design and, more importantly, in the property of global stability and explicit bounds on the convergence rate of the observer.

\begin{lemma}
  \label{lemma:boundedness}
  Consider the system~(\ref{eq:Kinematics}-\ref{eq:output}) and
  the above filter \eqref{filter}. If
  \begin{itemize}
  \item $v \in \mathbb{R}^n$ is bounded and such that the measured direction  $y$ is
 persistently exciting,
 \item $x$ is bounded and never crosses zero, and
 \item $a \in \rr^n$ is constant
\end{itemize}
 then $|\tilde{x}|$ (and hence $|\hat{x}_1|$) is uniformly bounded w.r.t. initial conditions and ultimately bounded by $\frac{1}{\gamma}|a|$.
\end{lemma}

\begin{proof}\label{x-1}
It is straightforward to verify that, in this case, the error-system equation is:
\begin{equation}\label{tilde-x}
\dot{\tilde{x}}_1 =a-k \pi_y \tilde{x}_1,
\end{equation}
whose general solution is:
\[\tilde{x}_1(t)=\Phi(t,0)\tilde{x}_1(0) +\int_{0}^{t}\Phi(\tau,t) a d \tau\]
Using \eqref{transition-bounds} it follows that $|\tilde{x}_1(t)| \leq (|\tilde{x}_1^0| +\frac{1}{\gamma}|a|)$ and $\varlimsup_{t\rightarrow +\infty}|\tilde{x}_1(t)| \leq \frac{1}{\gamma}|a|$. Since $x$ is bounded by definition, it follows that $\hat{x}_1$ is also bounded.
\end{proof}

For the design of an exponentially stable observer in the case where $a \neq 0$ the following two technical lemmas are instrumental.

\begin{lemma}
 \label{th:complementary_boundedness}
 Assume that $y \in S^{n-1}$ is persistently exciting. The matrix-valued function $M(t)$ solution to:
  \begin{equation}\label{M}
  \dot{M}=I -k \pi_y M, \;\;\;M(0)=M(0)^T=M^0>0
  \end{equation}
  is bounded and always invertible, and its condition number is bounded.
\end{lemma}
\begin{proof} See appendix \ref{app}.
\end{proof}

\begin{lemma} \label{persis2}
Assume that $y \in S^{n-1}$ is persistently exciting and $\dot{y}$ is uniformly continous.
The dual output $y^\star:=\frac{M^{-1}y}{|M^{-1}y|}$ is also persistently exciting.
\end{lemma}

\begin{proof} See appendix \ref{ystarexcitation}.
\end{proof}

\noindent The observer design presented hereafter is based on the
association of the filter \eqref{filter} that ensures, as we will show, that the variable $z:= \hat{x}_1 +Ma$ converges to $x$, with a second filter that provides an estimate $\hat{z}^{\star}$ of $z^\star:=M^{-1}z$. It then suffices to pre-multiply this second estimate by $M$ to obtain an estimate of $x$. The following theorem specifies the observer design and its convergence properties in the case where the output $y$ is persistently exciting.

\begin{theorem}
 Consider the system~(\ref{eq:Kinematics}-\ref{eq:output}) along with \eqref{filter} and \eqref{M}.
 Define the virtual observer $z$ as follows
 \begin{equation}
 z:=\hat{x}_1 +Ma
 \end{equation}
 and the dual observer $\hat{z}^\star$ of $z^\star:=M^{-1}z$ as follows
 \begin{equation}\label{x*}
\dot{\hat{z}}^\star=v^\star -k^\star \pi_{y^\star}\hat{z}^\star,\;\; \hat{z}^\star(0)=\hat{z}^{\star}_0
 \end{equation}
 with $v^\star:=M^{-1}\left(v-M^{-1}\hat{x}_1\right)$ a known term and $k^\star$ any positive gain. If $y$ is persistently exciting in the sense of Lemma \ref{lem:PE2} then
the virtual error $\tilde{x}_z:=x-z$, the position error $\tilde{x}=x-M\hat{x}^\star$, and the adaptation error $\hat{a}-a$ (with $\hat{a}:=\hat{z}^\star-M^{-1}\hat{x_1}$), globally exponentially converge to zero.
 \end{theorem}
\begin{proof}
The proof proceeds step by step. Concerning the convergence of $\tilde{x}_z$ to zero, one easily verifies, using  (\ref{filter}-\ref{M}), that:
\begin{equation}\label{dot-z}
\dot{z}=v+a +k \pi_y z
\end{equation}
Differentiating $\tilde{x}_z$, and using  \eqref{eq:Kinematics} and \eqref{dot-z}, one obtains:
\[
    \dot{\tilde{x}}_z =k \pi_y z=-k \pi_y \tilde{x}_z\]
This equation being the same as the one for $\tilde{x}_1$ in the case where $a=0$, one concludes as in Lemma \ref{th:complementary_convergence} that $\tilde{x}_z =0$ is uniformly globally exponentially stable, provided that $y$ is persistently exciting.

Concerning the convergence of $\tilde{x}$ to zero,
differentiating $z^\star:=M^{-1} z$, and using \eqref{filter} and \eqref{M}, yields:
\begin{equation}\label{dot-z*}
\dot{z}^\star=v^\star
\end{equation}
Now, differentiating $\tilde{z}^\star=z^\star-\hat{z}^\star$, and using \eqref{x*} and \eqref{dot-z*}, it comes that:
\begin{align*}
\dot{\tilde{z}}^\star &=k^\star \pi_{y^\star}\hat{x}^\star \\
&=-k^\star \pi_{y^\star}(\tilde{z}^\star-z^\star)
\end{align*}
Since $z^\star=M^{-1} z$, $\tilde{x}_z=x-z$, and $\pi_{y^\star}M^{-1}x \equiv 0$, one easily verifies that $\pi_{y^\star}z^\star=-\pi_{y^\star}M^{-1}\tilde{x}_z$. Therefore:
\begin{align*}
\dot{\tilde{z}}^\star&=-k^\star \pi_{y^\star} \tilde{z}^\star -k^\star \pi_{y^\star}M^{-1}\tilde{x}_z
\end{align*}
Since $y^\star$ is a persistently exciting (from Lemma \ref{persis2}), the above equation is similar to the one of $\tilde{x}_1$ in the case where $a=0$, except for the additive "perturbation" term $-k^\star \pi_{y^\star}M^{-1}\tilde{x}_z$ which converges exponentially to zero, due to the exponential convergence of $\tilde{x}_z$ to zero. It is immediate to show that this exponentially vanishing perturbation does not prevent $\tilde{z}^\star$ from globally converging to zero exponentially.
Since $\tilde{x}=M\tilde{z}^\star+\tilde{x}_z$, and since $\tilde{z}^\star$ and $\tilde{x}_z$ globally converge to zero exponentially, the position error $\tilde{x}$ also globally converge exponentially to zero.

Finally, using the definition of $\hat{a}$, one gets $a-\hat{a}=\tilde{z}^\star$ whose exponential convergence to zero has already been established.
\end{proof}

\section{Simulation}\label{sec:simulations}
We consider the example of a moving target point observed by a camera. The point moves in the 3D space ($n=3$) along a circular trajectory at a fixed altitude ($z=3m$) above the ground. The frame associated to the camera is located at the origin of the inertial frame whose optical axis is aligned with the
$z$-axis and looks up at the moving point. The measure $y \in S^2$ corresponds to the spherical projection of the point, given by the algebraic transformation $y=\frac{y^P}{|y^P|}$, where $y^P$ is the projective measure provided by the camera. The measurement of the velocity $v$ is biased by $a=(0.33,\;0.66,\;0.99)^\top$, and $v$ is chosen so that  $v+a=(-0.5\sin 0.5t, \; 0.5\cos 0.5t, \;0)^\top$. The following values of the observer gains are used: $k=0.5$ and $k^\star=5$.
\begin{figure}[h!]
\begin{center}
\includegraphics[scale=1.15]{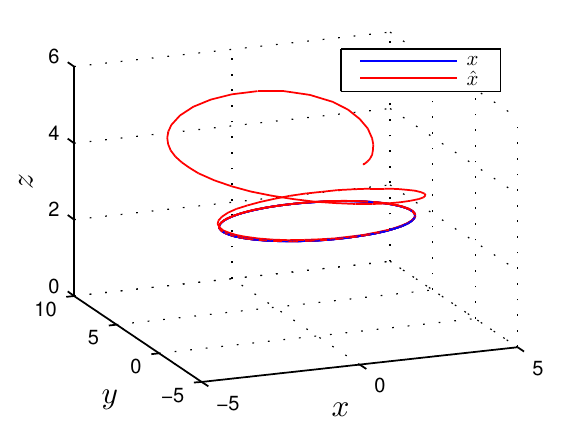}
\caption{Evolution of the system/observer pair in 3D space} \label{fig1}
\end{center}
\end{figure}
\begin{figure}[h!]
\begin{center}
\includegraphics[scale=1.15]{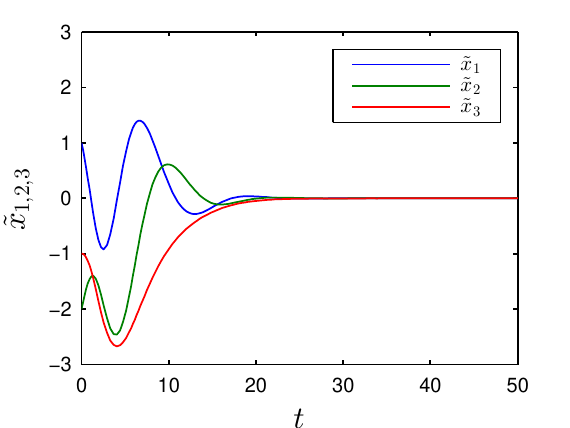}
\caption{Evolution of the observer error with respect to time} \label{fig2}
\end{center}
\end{figure}
\begin{figure}[h!]
\begin{center}
\includegraphics[scale=1.15]{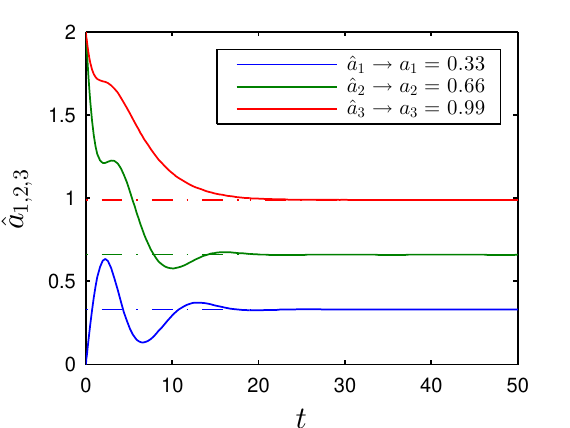}
\caption{Evolution of the estimate $\hat{a}$ with respect to time}\label{fig3}
\end{center}
\end{figure}

In Figures (\ref{fig1}-\ref{fig3}) the performance of the observer in the ideal noise-free case is shown. From these figures one can observe the exponential convergence of all estimation errors to zero. In figures (\ref{fig1b}-\ref{fig3b}), the observer algorithm is simulated in the case where the 3D bearing measurement $y$ is calculated from the position $x$ to which a uniform noise $w$ taking values in the interval $[-0.5m,0.5m]$ is added. Figures (\ref{fig1b}-\ref{fig3b}) show that the high frequency part of the noise is filtered by the proposed algorithm so that the performance of the proposed observer is not much reduced.
\begin{figure}[h!]
\begin{center}
\includegraphics[scale=1.15]{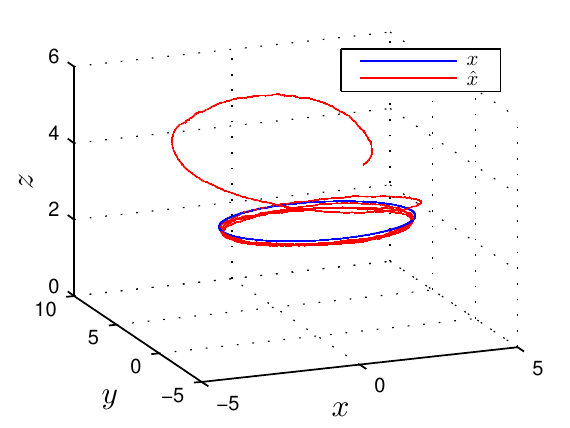}
\caption{Evolution of the system/observer pair when the output is affected by a noise}\label{fig1b}
\label{sim1}
\end{center}
\end{figure}
\begin{figure}[h!]
\begin{center}
\includegraphics[scale=1.15]{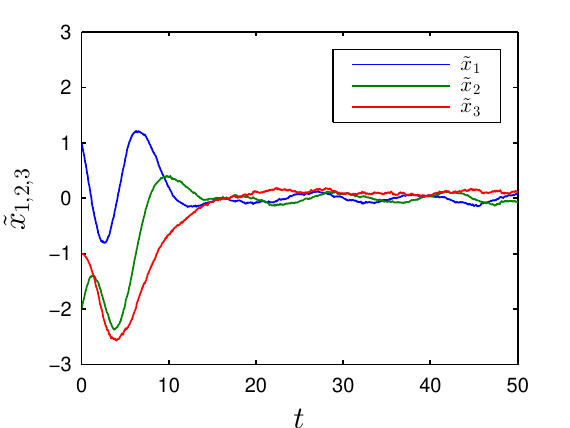}
\caption{Evolution of the observer error in the presence of measurement noise }\label{fig2b}
\label{sim1}
\end{center}
\end{figure}
\begin{figure}[h!]
\begin{center}
\includegraphics[scale=1.15]{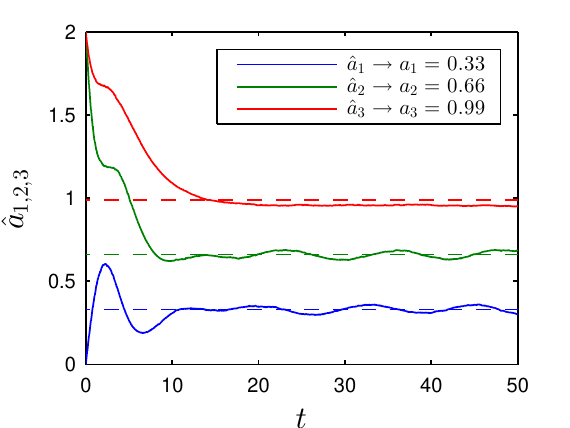}
\caption{Evolution of the estimate $\hat{a}$ in the presence of measurement noise}\label{fig3b}
\label{sim1}
\end{center}
\end{figure}
\vspace{-0.2cm}
\section{Concluding remarks}\label{sec:conclusions}
In this paper, we discussed the issue of observability of a moving object in $\rr^n$ from bearing measurement, proposed a nonlinear observer of the object's position in the case where the measured velocity of the object is biased, and carried out a detailed analysis of this observer in the case where the bearing measurement satisfies a condition of persistent excitation. There is an increasing number of emerging applications that can make use of such an observer. We think, in particular, of applications involving cameras for relative localization of mobile robot teams. We believe that extending the observer design methodology described in the paper to the estimation of the relative pose between to mobile objects evolving in $SE(n)$, with applications in $SE(3)$, is possible. This is one of future extensions of this work.

\vspace{-0.2cm}

\section{Acknowledgments}

This work was supported by the ANR-ASTRID project SCAR ``Sensory
Control of Unmanned Aerial Vehicles'', the ANR-Equipex project ''Robotex'' and the Australian Research Council through
the ARC Discovery Project DP120100316 "Geometric Observer Theory for
Mechanical Control Systems".

\appendix
\setcounter{section}{0}
\subsection{Proof of lemma \ref{lem:PE2}} \label{Plem:PE2}
Let us first show that \eqref{eq:persistence2} implies \eqref{eq:persistence4}.\\
For $\tau \in [t,t+\delta]$ one has
\[
\begin{array}{l}
|b^Ty(\tau)|^2 = |b^T(y(t)+\left( \begin{array}{c}\dot{y}_1(s_1) \\ \vdots \\ \dot{y}_n(s_n)\end{array}\right)(\tau-t)|^2\\
\geq |b^T(y(t)|^2-2n d(\delta)(\tau-t)-4n^2d(\delta)^2(\tau-t)^2
\end{array}
\]
for some $s_i\in[t,\tau]$ ($i=1,\ldots,n$) and $d(\delta)=\sup_{\tau \in [t,t+\delta]}  |\dot{y}(\tau)|$. Choose $b=y(t)$ so that $b^Ty(t)=1$, then
\[
|b^Ty(\tau)|^2 \geq 1-2n d(\delta)(\tau-t)-4n^2d(\delta)^2(\tau-t)^2
\]
and
\[
\int_t^{t+\delta} |b^Ty(\tau)|^2 \mathrm{d}\tau \geq \delta-nd(\delta) \delta^2-\frac{4}{3}n^2d(\delta)^2\delta^3
\]
Clearly there exists $\epsilon>0$ (independent of $t$) such that
\[
d(\tau)\leq \epsilon ~\Rightarrow ~ \int_t^{t+\delta} |b^Ty(\tau)|^2 \mathrm{d}\tau\geq \delta-\mu
\]
Let us proceed by contradiction and assume that \eqref{eq:persistence4} does not hold, i.e. $|\dot{y}(\tau)|<\epsilon$, $\forall t \in [t,t+\delta]$, then $d(\delta)<\epsilon$ and $\int_t^{t+\tau} |b^Ty(\tau)|^2 d\tau \geq \delta - \mu$. This contradicts \eqref{eq:persistence2} according to which
\[
\forall b \in S^{n-1}~:~\int_{t}^{t+\delta}|\pi_{y(\tau)}b|^2 \mathrm{d}\tau = \delta-\int_t^{t+\delta} |b^Ty(\tau)|^2 \mathrm{d}\tau \geq \mu
\]
Therefore \eqref{eq:persistence4} holds true.\\~\\
We now show that \eqref{eq:persistence4} implies \eqref{eq:persistence2}\\
Using the (assumed) uniform continuity of $\dot{y}$, \eqref{eq:persistence4} implies the existence of an interval $[t_1,t_2] \subset [t,t+\delta]$ such that $(t_2-t_1)=\epsilon_1(\epsilon)>0$ and $1.5>|y(t_2)-y(t_1)| \geq \epsilon_2(\epsilon):=\frac{\epsilon}{2n}\epsilon_1(\epsilon)>0$.
Now, $|b^Ty(t)|^2 =\cos(\theta(t,b))^2$ with $\theta(t,b)$ denoting the angle between $b$ and $y(t)$. The previous inequality in turn implies that $\cos(\theta(t_1,b))^2$ and $\cos(\theta(t_2,b))^2$ cannot be both equal to one. Therefore, $\cos(\theta(t_1,b))^2=1-\nu_1(\epsilon,b)$ and $\cos(\theta(t_2,b))^2=1-\nu_2(\epsilon,b)$, with $\max(\nu_1,\nu_2)\geq \nu(\epsilon,b) >0$. Since $\nu(\epsilon,.)$ is a continuous function depending on variables that take values in the compact set $S^{n-1}$ it reaches its bounds. This implies that $\nu(\epsilon,b)>\bar{\nu}(\epsilon)>0$, $\forall b\in S^{n-1}$. Using the fact that the uniform boundedness of $\dot{y}$ yields the uniform boundedness of $\dot{\theta}$ and thus of $\frac{d}{dt}\cos(\theta)^2$, this in turn implies that $\forall b \in S^{n-1}$:
\[
\int_{t_1}^{t_2} |b^Ty(\tau)|^2 \mathrm{d}\tau=\int_{t_1}^{t_2} \cos(\theta(\tau,b))^2 \mathrm{d}\tau\leq (t_2-t_1)(1-\epsilon_3(\epsilon))
\]
with $\epsilon_3(\epsilon)>0$. Therefore, $\forall b \in S^{n-1}$:
\[
\begin{array}{lll}
\int_t^{t+\delta} |b^Ty(\tau)|^2 \mathrm{d}\tau & = & \int_t^{t_1} |b^Ty(\tau)|^2 \mathrm{d}\tau+\int_{t_1}^{t_2} |b^Ty(\tau)|^2 \mathrm{d}\tau \\
~ & ~ &+\int_{t_2}^{t+\delta} |b^Ty(\tau)|^2 \mathrm{d}\tau\\
~ & \leq & (t_1-t)+(t_2-t_1)(1-\epsilon_3(\epsilon))\\
~ & ~ & +(t+\delta-t_2)\\
~ & \leq & \delta - \mu(\epsilon)
\end{array}
\]
with $\mu(\epsilon)=\epsilon_1(\epsilon)\epsilon_3(\epsilon)>0$.

\subsection{Proof of Lemma \ref{th:complementary_boundedness}} \label{app}
To prove that $M$ is bounded, is suffices to ensure that, for any constant vector $b \in S^{n-1}$, $|Mb|$ is bounded. Define $u=Mb$, it follows that:
\[\dot u= b-k\pi_y u\]
This equation is similar to the equation \eqref{tilde-x} of $\tilde{x}_1$. Therefore
\[
|M(t)b| \leq  |M(0)b| +\frac{1}{\gamma},\;\;\forall b \in S^{n-1}
\]
To show that $M$ is an invertible matrix, define $\Delta:=\det(M)$. From Jacobi's formula, one has
\begin{align}
\dot{\Delta} & = \Delta \tr(M^{-1}\dot{M}) \notag\\
&= \Delta \tr(M^{-1}-kM^{-1}\pi_yM)\notag\\
&= \Delta \tr(M^{-1})-k \Delta \tr(\pi_yMM^{-1})\notag\\
&= -k(n-1)\Delta + \Delta \tr(M^{-1}) \label{delta}
\end{align}
Note that this equation holds even if $M$ is not invertible. Indeed, using the fact that $\det(M)=\prod_{i=1}^{n}\lambda_i$  and $\tr(M^{-1})= \sum_{i=1}^{n} \frac{1}{\lambda_i}$, with $\lambda_i$ ($i=1 \ldots n$) the eigenvalues of $M$, one verifies that
\begin{equation} \label{trM}
\dot{\Delta}  =-k(n-1)\Delta+\sum_{i=1}^{n}\prod_{j=1,j\neq i}^{n}\lambda_j
\end{equation}
Since $M(0)$ is symmetric positive definite by assumption, all eigenvalues of $M(0)$ are positive and $\Delta(0)>0$.
Assume that $\Delta$ is equal to zero for the first time at the time instant $t_0>0$. Then, $\tr(M(t))>0$
on $[0,t_0)$ and $\tr(M(t))\geq 0$. In view of \eqref{trM}, $\Delta(t)\geq r(t)$ with $r(t)$ the solution to the equation $\dot{r}= -k(n-1)r$, with $r(0)=\Delta(0)$. Therefore $\Delta(t)\geq \Delta(0) \exp(-k(n-1)t)>0$, $\forall t$. This contradicts the existence of $t_0$ and proves that $M(t)$ is always invertible.\\
Let us now prove that $\Delta(t)$ is lower bounded by a positive number. Rewrite equation \eqref{delta} as follows
\begin{equation}\label{delta-1}
\dot{\Delta}  =-\left(\tr(M^{-1})-(n-1)k \right)\Delta,
\end{equation}
Using the fact that $\tr(M^{-1})>\frac{n}{\Delta^{1/n}}$, this equation shows that $\dot{\Delta} \geq 0$ if $\Delta<(\frac{n}{k(n-1)})^{1/n}$. Therefore $\Delta$ is ultimately lower bounded by $(\frac{n}{k(n-1)})^{1/n}$.

Finally, since $\Delta$ is lower bounded by a positive number and $M$ is upper bounded, it follows (by direct application of \cite{Piazza2002141}) that the condition number
\[\kappa(M)=||M||.||M^{-1}|| \leq \frac{2}{\Delta} \left (\frac{||M||_F}{\sqrt{n}}\right)^n\]
is upper bounded.

\subsection{Proof of Lemma \ref{persis2}} \label{ystarexcitation}
From the equation \eqref{M} of M one gets
\[
\frac{d}{dt}(My^{\star})=y^{\star}+M\dot{y}^{\star}
\]
Therefore
\[
(My^{\star})(t)=(My^{\star})(0)+\int_0^t (y^{\star}+My^{\star})(\tau) \mathrm{d}\tau
\]
and
\[
|(My^{\star})(t)|\geq \left |\int_0^t (y^{\star}+My^{\star})(\tau) \right |-|(My^{\star})(0)|
\]
Define:
\begin{itemize}
\item $k_M:=\sup_{t\in [0,+\infty[} ||M(t)||$, which implies that $|(My^{\star})(t)|\leq k_M$, $\forall t$,
\item $c(t):=\sup_{\tau\in [0,t]}|\dot{y}^{\star}(\tau)|$,
\item $w(t):=(y^{\star}+My^{\star})(t)$.
\end{itemize}
 One has
$w(\tau)=y^{\star}(0)+\left( \begin{array}{c} \dot{y}_1^{\star}(s_1)\\ \vdots \\ \dot{y}_1^{\star}(s_n) \end{array} \right) \tau + (My^{\star})(\tau)$ for some $s_i \in [0,\tau]$ ($i=1,\ldots,n$). Therefore
\[
\begin{array}{lll}
|\int_0^t w(\tau) \mathrm{d}\tau| & \geq &|\int_0^t  y^{\star}(0) \mathrm{d}\tau|-\left |\int_0^t \left( \begin{array}{c} \dot{y}_1^{\star}(s_1)\\ \vdots \\ \dot{y}_1^{\star}(s_n) \end{array} \right) \tau \mathrm{d}\tau \right| \\
~ & ~ & -|\int_0^t (My^{\star})(\tau)\mathrm{d}\tau|\\
~ & \geq & f(t,t_1) ~~;~t_1\geq t
\end{array}
\]
with $f(t,t_1):=t-nc(t_1)\frac{t^2}{2}-k_M c(t_1)t$. The function $f(.,t_1)$ monotically increases for $0 \leq t \leq t_2=\frac{1-k_Mc(t_1)}{nc(t_1)}$ if $c(t_1)<\frac{1}{k_M}$. Define $g(c):=t^{\star}(c)-nc\frac{t^{\star}(c)^2}{2}-k_M ct^{\star}(c)$ with $t^{\star}(c)=\frac{1-k_Mc}{nc}$. The function $g$ decreases monotically on $[0,\frac{1}{k_M}]$ with $g(0)=+\infty$ and $g(1/k_M)=0$. Let $\bar{c}>0$ denote the value of $c$ such that $g(\bar{c})=2k_M$, and set $t_1=t_2=\frac{1-k_M\bar{c}}{n\bar{c}}$. If $c(t_1)=\sup_{\tau\in [0,t_1]}|\dot{y}^{\star}(\tau)|<\bar{c}$ then
$|\int_0^{t_1} w(\tau) \mathrm{d}\tau|>2 k_M$. Since $|(My^{\star})(t_1)| \geq |\int_0^{t_1} w(\tau) \mathrm{d}\tau|-|(My^{\star})(0)|$ with $|(My^{\star})(0)|<k_M$, one concludes that
$|(My^{\star})(t_1)|>k_M$ (contradiction). Therefore, $\sup_{\tau\in [0,t_1]}|\dot{y}^{\star}(\tau)| \geq \bar{c}$, which proves the existence of a time-instant $t_3 \in [0,t_1]$ such that $|\dot{y}^{\star}(t_3)|\geq \bar{c}$.\\
The same proof repeated on every interval $[kt_1,(k+1)t_1]$ ($k\in \nn$) shows that $|\dot{y}^{\star}|$ is periodically larger than $\bar{c}>0$. This establishes that $\dot{y}^{\star}$ is persistently exciting.

\bibliographystyle{IEEEtran}
\bibliography{bibfile3}
\end{document}